\documentclass{ws-idaqp}
\usepackage{amsfonts,amsmath,amssymb,bm,color,blackboard}
\usepackage{graphicx}
\graphicspath{{figures/pdf/},{figures/eps/}}

\def\vec#1{\mathbf{#1}}

\def\ket#1{| #1 \rangle}
\def\bra#1{\langle #1 |}
\def\ip#1#2{\langle {#1} | {#2} \rangle}

\def\norm#1{\lVert #1 \rVert}
\def\dim{\operatorname{dim}}
\def\Tr{\operatorname{Tr}}
\def\dd{\operatorname{d}\!}
\def\UU{\mathbf{U}}
\def\SU{\mathbf{SU}}
\def\PU{\mathbf{PU}}
\def\sp{\mathfrak{sp}}
\def\su{\mathfrak{su}}
\def\uu{\mathfrak{u}}
\def\F{\mathfrak{F}}
\def\G{\mathfrak{G}}
\def\L{\mathfrak{L}}
\def\SS{\mathcal{S}}
\def\H{{\cal H}}

\def\N{{\cal N}}

\def\RR{\mathbb{R}}
\def\CC{\mathbb{C}}
\def\eps{\varepsilon}
\def\Im{\, \mathfrak{Im}}
\def\Re{\, \mathfrak{Re}}

\begin{document}

\title
      {A Closer Look at Quantum Control Landscapes \&
       their Implication for Control Optimization}

\author{Pierre de Fouquieres and Sophie G.~Schirmer} 

\address{College of Science, Physics, Swansea University, 
         Swansea, SA2 8PP, United Kingdom \\
         \email{sgs29@swan.ac.uk}}
\date{\today}

\maketitle

\begin{abstract}
The control landscape for various canonical quantum control problems is
considered.  For the class of pure-state transfer problems, analysis of
the fidelity as a functional over the unitary group reveals no
suboptimal attractive critical points (traps).  For the actual
optimization problem over controls in $L^2(0,T)$, however, there are
critical points for which the fidelity can assume any value in $(0,1)$,
critical points for which the second order analysis is inconclusive, and
traps.  For the class of unitary operator optimization problems analysis
of the fidelity over the unitary group shows that while there are no
traps over $\UU(N)$, traps already emerge when the domain is restricted
to the special unitary group.  The traps on the group can be eliminated
by modifying the performance index, corresponding to optimization over
the projective unitary group.  However, again, the set of critical
points for the actual optimization problem for controls in $L^2(0,T)$ is
larger and includes traps, some of which remain traps even when the
target time is allowed to vary.
\end{abstract}



\ccode{Add AMS classification here.}

\section{Introduction}

Quantum theory has been in existence for about a century but until
recently, the main emphasis in the field was on constructing Hamiltonian
models and solving the Schrodinger equation.  Although it was recognized
that external fields and potentials could change the Hamiltonian and
thus the dynamics of a system, and such external fields were certainly
used in many areas from nuclear magnetic resonance to atomic physics to
effect changes to the system, it was only rather recently that the full
potential of using such external fields was recognized.  Since then the
subject of control of quantum systems has developed from a niche area
into a subject of rapidly growing interest and importance, with an ever
increasing number of applications ranging from quantum chemistry to
quantum information processing (see e.g.\ Ref.~\refcite{rabitz2009}).
Quantum control has been applied, for instance, to influence the outcome
of chemical reactions\cite{somloi1993}, to prepare entangled
states\cite{xwang2009}, which are a resource in quantum metrology and
information processing, and to realize quantum
gates\cite{nigmatullin2009}, which are fundamental building blocks for a
quantum computer.  The scope of the applications is vast.

As the potential and importance of quantum control was realized,
attempts were made to formulate quantum control problems abstractly, and
establish solid mathematical foundations.  Some pioneering attempts were
necessarily imperfect and the theory is still under development, but
there have been a number of interesting theoretical results.  One
important area is the control landscape\cite{chakrabarti2007}.  While
it is not too difficult in many cases to formulate a particular task in
terms of a control optimization problem, i.e., of finding an admissible
control that maximizes a performance index, solving the resulting
control optimization problems is challenging, especially for large or
complicated systems, and in almost all cases they can only be solved
computationally using numerical optimization strategies.  Since the
evaluation of the performance index for quantum control problems
generally requires the solution of an operator or partial differential
equation, evaluation of the latter becomes computationally demanding and
efficient algorithms are very important.  This poses a challenge as we
are generally interested in finding global optima of a performance
index.  While it is possible to employ direct global search strategies
such as evolutionary algorithms, these are computationally expensive,
and success is generally not guaranteed.  On the other hand, there are
efficient algorithms for finding local optima of reasonably well-behaved
functions, especially when gradients can be computed relatively cheaply
compared to the computational overhead involved in evaluating the
performance index, as can be achieved in our case, including the Krotov
method\cite{somloi_controlled_1993,maday_new_2003,palao_quantum_2002},
concurrent-update gradient-ascent pulse engineering
(GRAPE)\cite{khaneja_optimal_2005} and various generalizations and
variants including quasi-Newton methods\cite{compare2010}.  This is
where the control landscape becomes relevant.

If the landscape is littered with many suboptimal, locally attractive
critical points then optimization strategies designed to find local
extrema have a very low probability of finding a global optimum from a
generic starting point.  In such a situation global search strategies
are usually required.  In the opposite extreme, if all local extrema are
actually global optima, then any algorithm designed to find local optima
will succeed in finding a global optimum from any generic starting
point.  In the intermediate case, if there are some sub-optimal (local)
extrema but their domain of attraction is small, then it is generally
still avantageous to employ efficient optimization strategies to find
local extrema.  The optimization may get trapped, i.e., attracted to
sub-optimal local extremum, but often simply restarting the algorithm
with a different initial condition will lead to success.  The success of
this strategy will depend on the probability of trapping, and usually
the ability to recognize and terminate trapped runs quickly.  The
landscape for quantum control problems was initially expected to be
complicated, but it was quickly realized that optimization algorithms
designed to find local extrema were often very successful at finding
controls for which the performance index to be optimized took values
very close to its \emph{global} optimum.  This has led a number of
researchers to study the issue and put forward arguments why this should
be the case, the core argument being that contrary to expectations, the
landscape for quantum control is extremely simple and devoid of
\emph{any} traps, i.e., suboptimal locally attractive critical points.

In this paper we revisit this topic, taking a closer look at the control
landscape.  Previous studies \cite{rabitz2004a,rabitz2004b,rabitz2005,
rabitz2006a,rabitz2006b,ho2006,hsieh2008a,hsieh2008b,wu2008,ho2009} have
mostly focused on regular critical points, in particular showing that
these can \emph{not} be traps for typical quantum control problems.
However, it is well-known in the broader field of optimal control that
these are not the only critical points\cite{montgomery1994}, and that
non-regular critical points must be considered\cite{bonnard1997,
chitour2006,chitour2008}.  For quantum control problems the existence of
non-regular critical points has only recently been
acknowledged\cite{wu2009} but their properties have not been studied
carefully.  For the Landau-Zehner it has recently been shown that there
are indeed no traps\cite{landau}, but in general it is not known if such
critical points can be traps.  The existence of such traps would
complicate the control landscape and invalidate some general results
arrived at by considering only regular points.  The main objective of
this paper is to show that non-regular critical points do exist for a
large class of quantum control systems and that at least some of these
are \emph{suboptimal extrema} for which the performance index can take
many values.  To this end we construct examples for which it can be
shown analytically that attractive local extrema exist.  This suggests
that further analysis is required to fully characterize the control
landscape, including the nature of potential traps and their practical
relevance, e.g., for the design of efficient algorithms.

The paper is organized as follows.  In Sec.~II, we formally define the
problem and give the mathematical prerequisites.  In Section III, we
take a closer look at the landscape for pure-state transfer problems.
In Sec.~IV we consider the landscape for unitary operator optimization
problems.  We conclude with a brief discussion in Sec.~VI.

\section{Background and Mathematical Prerequisites}

The focus of this paper is quantum systems subject to Hamiltonian
evolution governed by the Schrodinger equation.  Although the latter
generally applies to states evolving over an infinite-dimensional
Hilbert space, we restrict our attention to systems where the
\emph{Hilbert space} of interest $\H$ is \emph{finite-dimensional}.
However, our problems will be infinite-dimensional in that the
\emph{space of controls} is an \emph{infinite-dimensional} function
space.  This is the context for a considerable amount of work on control
of quantum systems, and can be partly motivated by the fact that in
practice one is usually either dealing with a system that intrinsically
has a finite number of degrees of freedom such as a collection of spins,
or for which the Hilbert space can be faithfully truncated and the
controlled dynamics restricted to finite-dimensional subspace of
interest.  Assuming this Hilbert space has dimension $N$, the evolution
of the system can be described by a unitary operator $U_f(t)\in\UU(N)$
satisfying the Schrodinger equation
\begin{equation}
\label{eq:SEU}
  i\hbar \dot{U}_f(t) = H_f(t) U_f(t), \quad U_f(0)=I,
\end{equation}
where $I$ is the identity matrix in $\UU(N)$ and $H_f(t)$ is the
Hamiltonian of the system, which in our case will be control-dependent,
as will be indicated by the subscript $f$.  For simplicity the control
dependence will be assumed to be linear in the (real-valued) controls
$f_m(t)$,
\begin{equation}
  H_f(t) = H_0 + \sum_{m=1}^M f_m(t) H_m.
\end{equation}
Here $H_0$ and $H_m$ are bounded Hermitian operators on $\H$, $H_0$
corresponding to the system's intrinsic Hamiltonian and $H_m$ modelling
the interaction with the control fields.  The control fields $f_m$ are
functions in a suitable function space.  In this paper we choose $f_m
\in \L^2(0,T)$, the space of square-integrable functions over the
interval $[0,T]\subset \RR$, although the results can be generalized to
other function spaces.  In the following we also choose $M=1$ unless
otherwise stated, choose units such that $\hbar=1$, and define the
two-point propagator $U_f(t_2,t_1)=U_f(t_2)U_f(t_1)^\dag$.

\subsection{Control problems}

Most typical control problems for (Hamiltonian) quantum systems fall in
one of the following three categories: pure-state transfer, observable
control and unitary operator implementation.  The pure-state transfer
problem consists of preparing the system in a desired state
$\ket{\Psi_g}$, usually at a given time $T$, assuming it is initialized
to some state $\ket{\Psi_0}$ at time $0$.  A natural measure of success
in achieving this goal is the transfer probability
\begin{equation} 
  \label{eq:FP}
  \F_P(f) = |\ip{\Psi_g}{\Psi(T)}|^2, 
\end{equation}
where in our case $\ket{\Psi(T)}=U_f(T)\ket{\Psi_0}$ and $U_f(T)$ is a
solution of~(\ref{eq:SEU}).  For a given Hamiltonian, initial and target
states, and fixed target time $T$, the transfer probability depends only
on the choice of control $f$, i.e., if the controls vary over $L^2(0,T)$
then $\F_P:L^2(0,T)\mapsto \RR$.  The observable control problem lies in
taking a system characterized by an initial density operator $\rho_0$ at
time $0$ to a state $\rho(T)$ at time $T$ for which the selected
observable $A$ attains its maximal possible value, i.e., we aim to
maximize the expectation value
\begin{equation}
  \label{eq:FD}
   \F_D(f) = \Tr[A \rho(T)],
\end{equation}
where $\rho(T)=U_f(T)\rho_0 U_f(T)^\dag$ and $U_f(T)$ is again a
solution of~(\ref{eq:SEU}).  The gate synthesis control problem consists
in choreographing the system dynamics to implement the unitary gate $V$
over the time interval $[0,T]$.  Using the Hilbert-Schmidt norm to
measure distance from the target gate $V$, this gives the gate fidelity
\begin{equation}
  \label{eq:FV}
  \F_V(f) = \frac{1}{N}\Re \Tr[V^{\dag} U_f(T)] 
\end{equation}
as a natural performance index, where $U_f(T)$ is a solution
of~(\ref{eq:SEU}) as before.

Broadly speaking, the objective of quantum control is to find a control
$f$ that maximizes one of these performance indices, where $f$ is
allowed to vary over a function space such as $L^2(0,T)$.  Here $T$ is
fixed unless stated otherwise.  In practice, quantum optimal control
problems can usually only be solved computationally using optimization
algorithms and such implementations require us to restrict the space of
controls to a finite-dimensional subspace, but in this paper we will be
chiefly interested in the ideal case where $f$ is allowed to vary over
$L^2(0,T)$.

More precisely, we wish to find the (global) maximum of a (non-linear)
functional $\F:L^2(0,T)\rightarrow\RR$, which in our case is either
$\F_P$, $\F_D$, or $\F_V$.  It will often be fruitful to think of $\F$
as a composition of the solution functional $U_f(T):L^2(0,T)\to \UU(N)$
taking a control field $f$ to the corresponding propagator at time $T$,
$U(T)$, via~(\ref{eq:SEU}), and the fidelity on the Lie group
$\G:\UU(N)\rightarrow\RR$, and we shall use
\begin{equation}
 \label{eq:fidG}
 \G_P(U)= |\bra{\Psi_g}U\ket{\Psi_0}|^2, \quad
 \G_D(U)= \Tr[AU\rho_0 U^\dag], \quad
 \G_V(U)= \frac{1}{N}\Re\Tr[V^\dag U],
\end{equation}
respectively, to denote the equivalents of $\F_P$, $\F_D$ and $\F_V$
above.

\subsection{Controllability considerations}

Whether we allow the fields $f_m$ to range over the whole of $L^2(0,T)$
or a smaller space such as the continuous or piecewise constant
functions, we must consider whether the aforementioned control problems
are feasible, i.e., whether the maximum of the performance index is
actually attainable.  In general, reachability of a target state that
globally optimizes the performance index is difficult to assess, but we
know that a sufficient condition for attainability of the global maximum
is controllability.  Although this is generally a much stronger
requirement, controllability is useful as it is well characterized.
Specifically, we say that a given system $H_f(t)$ is controllable if for
all endpoint conditions -- either $\{\ket{\Psi_0},\ket{\Psi_g}\}$,
$\{\rho_0,A\}$ or $\{I,V\}$, there exists a time $T$ and an (admissible)
control $\vec{f}$ that solves the corresponding control problem.  In
general, for bilinear control systems on a Lie group, as in our case,
this problem reduces to one about the Lie algebra $\L$ generated by the
matrices $\{iH_m\}_{m=0}^M$.  For pure-state control problems
controllability is equivalent to $\L$ being a representation of either
$\sp(N/2)$, $\sp(N/2)\oplus \uu(1)$, $\su(N)$ or $\uu(N)$.  For
mixed-state (density-operator) control problems, we require $\su(N)$ or
$\uu(N)$, while the system is fully unitary operator controllable only
when $\L=\uu(N)$\cite{albertini2001,schirmer2002}, although the
distinction between $\su(N)$ and $\uu(N)$ is artificial in most cases as
the extra dimension contributes only a global phase, which is generally
not observable.  Since we usually wish to fix the final time $T$ when
carrying out control optimization, we are also interested in the
stronger notion of exact-time controllability, namely, whether a system
has some critical time $T_c$ such that all endpoint conditions give
solvable control problems for any $T>T_c$.  It turns out that,
neglecting global phase, these notions of controllability are equivalent
to the previous three by Theorem 13 of \cite{jurdjevic1997}.

Apart from these now standard results, we will need the following:
\begin{theorem}
Let $H=H_0+f(t)H_1$ be a control system, $\L=\uu(N)$ or $\su(N)$ and
$G=\UU(N)$ or $\SU(N)$ be associated Lie group. 

If the set of all commutator expressions in $i H_0$ and $i H_1$ joined
with $i H_0$ and $i H_1$ span $\L$ (Lie algebra rank condition), then
there exists a maximal time $T_{\max}$ and neighborhood $\N$ of
$(H_0,H_1)$ in $i\L\times i\L$ such that for all systems $s\in\N$ and
elements $g\in G$, there is a control taking $s$ to $g$ in time
$T<T_{\max}$.
  
If the set of all commutator expressions in $i H_0$ and $i H_1$ joined
with $iH_1$ span $\L$ (exact-time Lie algebra rank
condition~\footnote{Note again that the exact-time and usual Lie algebra
conditions match for $\su(N)$.}), then there is a critical time $T_c$
and neighborhood $\N$ of $(H_0,H_1)\in i\L \times i\L$ such that for all
$s\in\N$, $g\in G$ and $T>T_c$, there is a control taking $s$ to $g$ in
time $T$.
\end{theorem}

\begin{proof}
If system specified by $(H_0,H_1)$ satisfies the Lie algebra rank
condition, by Theorem 1 of Ref.~\refcite{jurdjevic1997} there is a set of values
$v_1,\ldots,v_n$ such that the function
\begin{align*}
 F :(t_1, \ldots, t_n) \mapsto
  e^{-i(H_0+v_n H_1) t_n} \cdots e^{-i( H_0 + v_1 H_1) t_1}
\end{align*}
has Jacobian $\dd F$ of full rank at some point $\vec{t}$. 
This implies that the extended function 
\begin{align*}
  F': (A,B,t_1,\ldots,t_n) \mapsto 
  (A,B, e^{-i(A+ v_n B) t_n} \cdots e^{-i( A + v_1 B) t_1})
\end{align*} 
also has Jacobian 
\(
\dd F' = \left(\begin{smallmatrix}
    I & 0 & 0\\
    0 & I & 0\\
    \ast & \ast & \dd F
\end{smallmatrix}\right)
\) of full rank at the point $(H_0, H_1,\vec{t})$.  Thus by the inverse
function theorem, choosing neighborhoods $\N$ of $(H_0,H_1)$ and $V$ of
$\vec{t}$ bounded, there is a neighborhood $W$ of $F(\vec{t})$ such
that $F'(\{A\} \times \{B\} \times V)$ includes $\{A\}\times\{B\} \times
W$ for every $(A,B)\in\N$.  By compactness (and connectedness), some
power $W^k$ covers the entire group $G$.  Letting $T$ be the supremum of
$\sum_n t_n$ over $V$, we see that for all systems in $\N$ the entire
group is reachable in less than time $kT$.
  
Assuming $(H_0,H_1)$ satisfies the exact-time Lie algebra rank condition, 
the same argument may be repeated with 
\begin{align*}
  F:(t_0,\ldots,t_n) \mapsto 
 e^{-i(H_0+v_n H_1)t_n} \cdots e^{-i(H_0+v_0 H_1) t_0}
\end{align*}
as constructed in the proof of Theorem 3 in
\cite{jurdjevic1997}, where now $F$ has domain restricted
to $\sum_n t_n = T$. This alteration shows that the critical time $T_c$
of exact-time controllability can indeed be chosen uniformly about
$(H_0,H_1)$, as claimed.
\end{proof}

\subsection{Gradient and Hessian formulas}

Finding the global optimum of a function is generally a very difficult
task.  On the other hand, there are many efficient algorithms to find
local extrema, i.e., attractive critical points of a function.  For this
reason we are particularly interested in the nature of the critical
points of the various functionals $\F$, which depends only on certain
local properties of the solution functional $U_f (t)$. To study the
latter, we need the identity
\begin{equation}
   U_{f+{\Delta}f}(t)-U_f(t) = -i\int_{t_0}^t 
   U_f(t,\tau)[H_{f+{\Delta}f}(\tau)-H_f(\tau)]U_{f+{\Delta}f}(\tau)\,\dd\tau
\end{equation}
which can be verified by differentiating both sides, and used to derive
\begin{multline}
  \label{eq:pert}
  U_{f+{\Delta}f}(T)
   = U_f(T) -i\int_{0}^T U_f(T,\tau)\Delta H(\tau)U_f(\tau)\dd\tau\\
     - \int_0^T\!\!\int_0^\tau U_f(T,\tau)\Delta H(\tau)
       U_f(\tau,\sigma){\Delta}H(\sigma)U_{f}(\sigma) 
       \dd\sigma\dd\tau + O(\norm{\Delta f}^3)
\end{multline}
where $\Delta H(\tau)=H_{f+\Delta f}(\tau)-H_f(\tau)$ is just $\Delta
f(\tau) H_1$ in the $M=1$ case.  For $\alpha,\beta\in L^2(0,T)$ we
define the linear map
\begin{equation}
  \label{eq:grd}
  \alpha \mapsto 
  \int_0^T \frac{\delta U_f(T)}{\delta f(\tau)} \alpha(\tau) \dd\tau, \qquad
  \frac{\delta U_f(T)}{\delta f(\tau)}= -iU_f(T,\tau)H_1 U_f(\tau),
\end{equation}
corresponding to the gradient of $U_f(T)$ at $f$, and the bilinear map
\begin{equation}
 \label{eq:hess}
  (\alpha,\beta) \mapsto
  - \!\!\!\underset{0<\sigma<\tau<T}{\int\int}\!\!\! U_f(T,\tau)\Delta H(\tau)
       U_f(\tau,\sigma){\Delta}H(\sigma)U_{f}(\sigma) 
       [\alpha(\tau)\beta(\sigma)+\beta(\tau)\alpha(\sigma)]\dd\sigma\dd\tau 
\end{equation}
corresponding to its Hessian.  Indeed, the second order perturbative
expansion (\ref{eq:pert}) is the functional analog of the Taylor
expansion of a function, with the gradient and Hessian functionals
corresponding to the first and second order terms, respectively.  Since
the Hessian is naturally a linear operator mapping $L^2(0,T)$ to itself,
it is convenient to have notation for referring directly to some such
operators, along with the elements of $L^2$ themselves.  We shall use
$\Pi$ to denote the unnormalized projection onto its argument, and
$\bullet$ as a placeholder for the argument of the anonymous function it
is a constituent of.  For example, $\Pi[\cos(\omega\bullet)]$ refers to
the projector onto the function $t\mapsto \cos(\omega t)$, that is, the
operator that maps each $\alpha \in L^2(0,T)$ to $\cos(\omega\bullet)\int_0^T
\cos(\omega\tau) \alpha(\tau)\dd\tau$ in $L^2(0,T)$.

The importance of the gradient and the Hessian functionals is that the
former determines the critical points of the functional $\F$, and the
latter their nature.  In particular, if the Hessian at a critical point
$f$ is negative definite, then the critical point is attractive and
corresponds to a local maximum over any finite-dimensional subspace of
$L^2(0,T)$; similarly, if it is positive definite, the critical point is
repulsive and corresponds to a local minimum over any finite-dimensional
subspace, and if the Hessian at the critical point can take both
positive and negative values, the critical point is a saddle.

\section{Pure-state landscape}

\subsection{All Critical Points over $\SU(N)$ or $\UU(N)$ global extrema?}

As mentioned in the introduction the problem of finding a control $f\in
L^2(0,T)$ to maximize the pure-state fidelity functional~(\ref{eq:FP})
can be viewed a composition of finding a $U_*\in \UU(N)$ or $\SU(N)$
that maximizes $\F_P$, and an optimization problem of finding a control
in $L^2(0,T)$ that realizes $U_f(T)=U_*$.  Defining
$A=\ket{\Psi_g}\bra{\Psi_g}$ and $\rho_0=\ket{\Psi_0}\bra{\Psi_0}$, we
have $|\bra{\Psi_g}U\ket{\Psi_0}|^2=\Tr[AU\rho_0 U^\dag]$, i.e., we can
rewrite the pure-state optimization problem as a general observable
optimization problem.  It is easy to show that a necessary and
sufficient condition for $U\in\UU(N)$ or $\SU(N)$ to be a critical point
of $\Tr[AU\rho_0U^\dag]$ is that $A$ and $U\rho_0 U^\dag$ commute, i.e.,
$[A,U\rho_0U^\dag]=0$.  In our case, as $A$ and $\rho_0$ (and thus
$U\rho_0 U^\dag$) are projectors onto pure states, there are only two
types of critical points: (i) $A$ and $U\rho_0 U^\dag$ are projectors
onto orthogonal subspaces of $\H$, in which case we have $\Tr[A U\rho_0
U^\dag]=0$, and (ii) $A$ and $U\rho_0 U^\dag$ are projectors onto the
same 1D subspace of $\H$, in which case we have $\Tr[A U\rho_0
U^\dag]=1$.  Thus, the landscape for pure-state optimization over
$\UU(N)$ or $\SU(N)$ is very simple: there are only two types of
critical points, corresponding to extremal values of the fidelity, i.e.,
global extrema, and no saddles or other critical points for which the
fidelity assumes values in $(0,1)$.  The critical points over $\SU(N)$
or $\UU(N)$ are sometimes referred to as \emph{kinematic} critical
points\cite{wu2009}.

Since then several papers have attempted to show that this result
extends to the actual optimization problem over $L^2(0,T)$, i.e., that
the fidelity $\F_P$ as a functional over $L^2(0,T)$ only has critical
points for which it achieves either its global minimum $0$ or maximum
$1$.  However, the arguments put forward in Ref.~\refcite{rabitz2004a}
and \refcite{rabitz2004b} rely, without rigourous justification, on the
property that the solution functional $f\mapsto U_f(T)$ is of full rank
\emph{everywhere}.  Ref.~\refcite{hsieh2008b} gives another proof but
relies on the same assumption.  Ref.~\refcite{rabitz2006a} argues that a
certain sequence of expressions starting with equation (14) and (15)
therein generate the full Lie algebra generated by $iH_0$ and $iH_1$,
and hence, assuming controllability, span the entire Lie algebra
$\su(N)$.  But this procedure only allows the generation of specific
linear combinations of commutators and, although it may seem plausible,
these are not guaranteed to span $\su(N)$.  For the more general density
matrix control problem, Ref.~\refcite{rabitz2006a} similarly assumes
that the map $f\mapsto U_f(T)$ is of full rank everywhere to derive that
the final state $U_f(T) \rho_0 U_f (T)^{\dag}$ must commute with the
observable $A$ at critical points, and hence that $0$ and $1$ are the
only critical values for pure states. Ref.~\refcite{ho2006} makes use of
the same assumption, referring to Ref.~\refcite{rabitz2006a} as evidence
of its validity, to re-derive this critical point characterization,
along with the result that all but the global maxima and minima are
saddle points.  Ref.~\refcite{wu2008} recognises the possibility of
singular controls $f$, here of $f\mapsto \Psi_f(T)$ not being full rank
for a particular control $f$, along with the issues this raises, but
defers analysis of such points to future work.  Ref.~\refcite{wu2009}
gives some characterization of singular controls and acknowledges that
these can be critical points of the performance index that do not
correspond to kinematical critical points, but considers only one
example involving eigenstates of a four-level system, to conclude based
on numerical simulations that singular controls do not appear to be
traps for this problem.  Therefore, the question of the existence of
non-kinematic critical points and their nature for the pure-state
fidelity~(\ref{eq:FP}), i.e., critical points $f \in L^2(0,T)$ for which
$\F_P$ does not assume extremal values, remains open.  Furthermore, for
the observable optimization problem (\ref{eq:FD}) it has been shown that
critical points exist that are traps at least to second order, i.e., for
which the Hessian is negative semi-definite but the fidelity does not
assume its global maximum value\cite{pechen2011}.

In this work we show that non-kinematical critical points exist for a
large class of pure-state control problems.  Since the mapping $f\mapsto
\Psi_f(T)$ is \emph{not} of full rank at these singular points, it shows
that the full-rank-everywhere hypothesis made in many of the works cited
above is generally not satisfied.  So even if the landscape was indeed
guaranteed to be trap-free under this hypothesis, we cannot draw any
conclusions from these results about the existence of traps for quantum
control problems.  This does not immediately imply the existence of
traps in the control landscape as these singular critical points need
not be attractive.  However, we further show that one can systematically
construct examples of non-kinematic critical points at which the Hessien
shows is \emph{strictly negative definite}.  This is possible as the
Hessian need not have finite rank at singular critical points, and thus
can have infinitely many negative eigenvalues.  Thus, not only do
non-kinematic critical points exist but they can be \emph{attractive},
i.e., traps.  The implication of these results is that any landscape
analysis based on considering only regular points is incomplete and
inconclusive, and the existence of examples for which traps provably
exist shows that the landscape is \emph{not universally trap-free}.

\subsection{Any critical value possible for critical points over $L^2(0,T)$?}

Let us start by considering, for any choice of final time $T$, a pure
state control problem for a two-level system of the general form:
\begin{subequations} 
\label{genpure}
\begin{align} 
  H_0 = \begin{pmatrix}
    a & 0\\
    0 & b
  \end{pmatrix} &, \quad H_1 = \begin{pmatrix}
    c & d\\
    \overline{d} & c
  \end{pmatrix} \\
  \bra{\Psi_g} = \frac{1}{\sqrt{2}} 
   [1 , e^{i \phi} ] e^{i T H_0} &, \quad 
   \ket{\Psi_0} = \frac{1}{\sqrt{2}} 
   \begin{bmatrix} e^{i \theta} \\ e^{- i \left( \theta + \phi \right)}\end{bmatrix}, 
\end{align}
\end{subequations}
where we can exclude the degenerate case $\ket{\Psi_g}=\ket{\Psi_0}$ by
choosing $\theta \not\equiv \frac{b-a}{2} T \pmod{ \pi }$.  For the
field $f \equiv 0$, which is identically zero over $[0,T]$, we have
$\bra{\Psi_g}U_f(T)\ket{\Psi_0} = \cos(\theta)$ and thus the pure-state
transfer fidelity $\F_P=|\bra{\Psi_g}U_f(T)\ket{\Psi_0}|^2 =\cos^2 (\theta)$.
The gradient is
\begin{align*}
  \frac{\delta\F_P}{\delta f(t)}
&= \bra{\Psi_g} \tfrac{\delta U_f(T)}{\delta f(t)} \ket{\Psi_0} \overline{\bra{\Psi_g}U_f(T)\ket{\Psi_0}} 
  + \bra{\Psi_g} U_f(T) \ket{\Psi_0} \overline{\bra{\Psi_g}\tfrac{\delta U_f(T)}{\delta f(t)}\ket{\Psi_0}} \\
&= 2\Re \left[\bra{\Psi_g} \tfrac{\delta U_f(T)}{\delta f(t)} \ket{\Psi_0} \overline{\bra{\Psi_g}U_f(T)\ket{\Psi_0}} \right]
\end{align*}
and inserting $\tfrac{\delta U_f}{\delta f(t)}= -i U_f(T)U_f(t)^\dag H_1 U_f(t)$
from (\ref{eq:grd}) this gives
\begin{equation*}
 \frac{\delta \F_P}{\delta f(t)} 
  = 2 \Im \left(\langle \Psi_g |U_f(T)U_f(t)^{\dag} H_1 U_f(t) | \Psi_0 \rangle 
   \overline{\langle \Psi_g |U_f(T)|\Psi_0 \rangle} \right)
\end{equation*}
For $f \equiv 0$ this gradient formula evaluates to
\begin{equation*}
 \Im \left( c e^{i \theta} + d e^{i ( a - b ) t - i ( \theta
   + \phi )} + \text{c.c.} \right) \cos ( \theta ) = 0
   \quad \text{ for all } t,
\end{equation*}
showing that $f=0$ is a critical point.  A value of $\theta$ can be
chosen to achieve any fidelity in the interval $(0,1)$, and this can be
done avoiding the degenerate case by choosing the right sign for
$\theta$ when necessary.  So we already see there is a large family of
control problems admitting critical points for which the fidelity
assumes any possible value but this is simply a special case of the
following general theorem.

\begin{theorem}
\label{thm:1} Given any bilinear control system with control Hamiltonian
$H_0 + f (t) H_1$, any target time $T>0$, and any $F \in (0,1)$, there
exist pairs of initial and target states $\ket{\Psi_0}, \ket{\Psi_g}$
and a control $f$ such that $f$ is a critical point of the fidelity
$\F_P$ with critical value $\F_P = F$.
\end{theorem}

\begin{proof}
Any bilinear control system $H_0+f(t)H_1$ with $f(t)=\tilde{f}(t)+\mu$
is equivalent to one with Hamiltonians $\tilde{H}_0 = H_0+\mu H_1, H_1$
and control $\tilde{f}(t)$.

Case 1: Suppose there is a value of $\mu\in\RR$ such that $\tilde{H}_0$
has degenerate eigenvalues.  Choose a two-dimensional subspace $\SS$ of
this eigenspace.  Then, restricted to the subspace $\SS$ the Hamiltonians
take the form
\begin{equation} 
   \label{unconpure}
    \tilde{H}_0 = \begin{pmatrix}
      a & 0\\
      0 & a
    \end{pmatrix}, \quad 
   H_1 = \begin{pmatrix}
      b & 0\\
      0 & c
    \end{pmatrix}. 
\end{equation}
It easy to see that $\tilde{f}\equiv 0$ is a critical point with
fidelity $\cos^2(\theta-\phi)$ for any $T$ for the state transfer
problem $\ket{\Psi_0} \mapsto \ket{\Psi_g}$ on $\SS$ with
\begin{equation*}
\bra{\Psi_g} =
    [\cos (\theta),\sin (\theta) e^{- i \gamma}], \quad
\ket{\Psi_0} = \begin{bmatrix} \cos (\phi) \\ \sin (\phi) e^{i \gamma} \end{bmatrix}
\end{equation*}
if the Hamiltonians restricted to $\SS$ take the form~(\ref{unconpure}).
Since $\SS$ is an invariant space of $U(t)=e^{-it\tilde{H}_0}$, if the
initial and final states are in $\SS$, the formulas for the fidelity and
gradient simplify to their analogues restricted to $\SS$.  Therefore,
$\tilde{f}=0$ is a critical point for the state transfer problem
$\ket{\Psi_0} \mapsto \ket{\Psi_g}$ for the original system if we embed
the initial and target state into the full Hilbert space $\H$ in the
obvious manner.

Case 2: If the eigenvalues of $H_0+\mu H_1$ are always distinct then we
can continuously parametrize the (unique) eigenvectors $v^{\max}(\mu)$
and $v_{\min}(\mu)$ corresponding to the maximal and minimal eigenvalues
of $H(\mu)=\tfrac{1}{|\mu|+1}(H_0+\mu H_1)$, respectively.  Letting
$\lambda^{\max}$, $\lambda_{\min}$ be the maximal and minimal
eigenvalues of $H_1$, we have the identities
\begin{align*}
 \underset{\mu\to-\infty}{\lim} \bra{v^{\max}(\mu)}(-H_1 )\ket{v^{\max}(\mu)}
  &= -\lambda_{\min} & \underset{\mu\to\infty}{\lim}
     \bra{v^{\max}(\mu)} H_1 \ket{v^{\max}(\mu)}
  &= \lambda^{\max} \\ \underset{\mu\to-\infty}{\lim}
     \bra{v_{\min}(\mu)}(-H_1) \ket{v_{\min}(\mu)}
  &= -\lambda^{\max} & \underset{\mu\to\infty}{\lim} 
     \bra{v_{\min}(\mu)} H_1 \ket{v_{\min}(\mu)} 
  &= \lambda_{\min}
\end{align*}
Hence by continuity we must have
$\bra{v^{\max}(\mu)}H_1\ket{v^{\max}(\mu)} = \bra{v_{\min}(\mu)} H_1
\ket{v_{\min}(\mu)}$ for some (finite) $\mu$, unless $\lambda^{\max} =
\lambda_{\min}$, which would mean that $H_1$ is a multiple of the
identity and so any value of $\mu$ would do.  Using this $\mu$ leads to
an effective system $\tilde{H}_0, H_1$ which, when restricted to the
subspace $\SS$ spanned by $\{v^{\max}(\mu), v_{\max}(\mu)\}$, and
expressed in this basis, is precisely of the form (\ref{genpure}).
Again we can extend the initial and target states in (\ref{genpure}) to
the full Hilbert space in a trivial way so that $\tilde{f}\equiv 0$ is a
critical point with fidelity $\cos^2\theta$ for the resulting state
transfer problem with the Hamiltonian $\tilde{H}_0+\tilde{f}H_1$.
Restriction to $\SS$ is again a well-behaved procedure since the span of
each eigenvector is invariant under the time evolution generated by the
corresponding Hamiltonian $\tilde{H}_0$.
\end{proof}

This theorem shows that for any control system, in particular any
controllable one, there exist pairs of initial and target states for
which we can achieve critical values of the fidelity in $(0,1)$.  The
initial and target states which exhibit such critical points in the
proof above are generally not eigenstates of $H_0$ and one might wonder
whether critical points for which the fidelity does not achieve extremal
values might not exist if the initial and target states are restricted
to eigenstates of the system Hamiltonian $H_0$, as is the case in many
applications, but even in this case there are counter-examples.

\begin{example}{\bf Non-extremal critical points for eigenstate transfer.}
Consider state transfer from $\ket{\Psi_0}$ to $\ket{\Psi_g}$ for
the system $H_f(t)=H_0+f(t) H_1$, where
\begin{align*}
  H_0 = \begin{pmatrix} 1 & 0 & 0 \\ 0 & 2 & 0 \\ 0 & 0 & 4
	\end{pmatrix}, \quad
  H_1 = \begin{pmatrix} 1 & \sqrt{\tfrac{2}{3}} & 0 \\
                        \sqrt{\tfrac{2}{3}} & 2 & \sqrt{\tfrac{1}{3}} \\
                         0 & \sqrt{\tfrac{1}{3}} & 4
        \end{pmatrix}, \quad
 \ket{\Psi_0}= \begin{bmatrix} 1 \\ 0 \\ 0 \end{bmatrix}, \quad
 \ket{\Psi_g}= \begin{bmatrix} 0 \\ 0 \\ 1\end{bmatrix}.
\end{align*}
The initial and target states are eigenstates of the system Hamiltonian
and the system is controllable as $H_0$ has distinct transition
frequencies and $H_1$ is connected.  However, for $T=2\pi$, for example,
the transfer fidelity $\F_P$ has a critical point at $f \equiv -1$ as
$\tfrac{\delta\F_P}{\delta f}= \tfrac{8}{9}\sin^2(\tfrac{T}{2})\sin(T)=0$ with fidelity
$\F_P=\frac{8}{9}\sin^4(\tfrac{T}{2})=\frac{8}{9}<1$.
\end{example}

\subsection{Attractive suboptimal critical points over $L^2(0,T)$}

We can furthermore show that the critical points above, for which the
fidelity does not assume extremal values, are not necessarily saddle
points either.  We show this for pure state examples but the problems
can be reformulated as observable optimization problems with
$A=\ket{\Psi_g}\bra{\Psi_g}$ for a density matrix that is a projector
onto the pure state $\ket{\Psi(t)}$.  Hence the examples also show the
existence of non-kinematic critical points for which the observable
fidelity $\F_D=\Tr[A\rho(T)]$ assumes non-extremal values, which are not
saddle points.

\begin{example}{\bf Critical saddle points with negative semi-definite Hessian.}
Given any target time $T$ and $H_f(t)=H_0+f(t)H_1$, choose
\begin{align*}
  H_0 = \begin{pmatrix} 
	 2 & 0 & 0 & 0\\ 
	 0 & 4 & 0 & 0\\ 
	 0 & 0 & 5 & 0\\
	 0 & 0 & 0 & 9
	\end{pmatrix}, \quad
  H_1 = \begin{pmatrix} 
	 0 & 1 & 0 & 0\\ 
	 1 & 0 & 1 & 0\\
         0 & 1 & 0 & 1\\ 
	 0 & 0 & 1 & 0
        \end{pmatrix}, \quad
  \ket{\Psi_0}= \begin{bmatrix}
                \cos\phi \\ 0 \\ 0 \\ \sin\phi
                \end{bmatrix}, \quad
  \ket{\Psi_g}= \begin{bmatrix} 
                e^{-i2T}\cos\theta \\ 0 \\ 0 \\ e^{-i9T}\sin\theta
                \end{bmatrix}.
\end{align*}
$f\equiv 0$ is a critical point of the fidelity as it can easily be
verified that $\bra{\Psi_g}U_0(T,t)H_1 U_0(t)\ket{\Psi_0}=0$.  The
corresponding critical value $\F_P=\cos^2(\theta-\phi)<1$ unless
$\theta=\phi \mod \pi$. The Hessian of the fidelity at this point can be
computed as
\begin{multline*}
-2 \cos (\theta-\phi) \!\!\!\underset{0<s<t<T}{\int\!\int}\!\! 
 [\cos\theta\cos\phi \cos(2(s-t)) 
 + \sin \theta \sin \phi \cos(4(s-t)) ] \gamma(s,t) \dd s \dd t
\end{multline*}
where $\gamma(s,t) = \alpha(s)\beta(t)+ \beta(s) \alpha(t)$ accounts for
the necessary symmetrization, and can be written as two instances of the
operator $C$ defined in the appendix
\begin{multline*}
  -2 \cos(\theta-\phi)
  \Big[\cos\theta\cos\phi\big(\Pi[\cos(2\bullet)]+\Pi[\sin(2\bullet)]\big)\\
      +\sin\theta\sin\phi\big(\Pi[\cos(4\bullet)]+\Pi[\sin(4\bullet)]\big)\Big].
\end{multline*}
This expression is negative semi-definite whenever $\cos\theta\cos\phi$
and $\sin\theta\sin\phi$ are both positive, equivalently, when $\theta$
and $\phi$ lie in the interior of the same quadrant.  For example, for
$\theta=\frac{\pi}{6}$ and $\phi=\frac{\pi}{3}$ we obtain
$\F_P=\frac{3}{4}$ and Hessian taking the form $-\frac{3}{4}
\big[\Pi[\cos(2\bullet)]+\Pi[\sin(2\bullet)] +\Pi[\cos(4\bullet)]
+\Pi[\sin(4\bullet)]\big]$.  Thus $f$ is a critical point, which is a
trap to second order.  This means that there exists a neighborhood of
the critical point from which no points can escape if we consider only
the second-order perturbative expansion of the fidelity around the
critical point.  However, such second-order traps are usually
third-order saddles as on the subspace where the Hessian vanishes the
local dynamics is determined by the first non-vanishing term and the
third order term usually has indefinite sign.  Indeed, by extending
(\ref{eq:pert}), we can compute the third derivative of $\F_P$
\begin{multline*}
12 \cos ( \theta - \phi ) 
\underset{0 < r < s < t < T}{\int \int \int} 
 [ \cos\theta \sin\phi \sin(4r+s+2t) - \sin\theta \cos\phi \sin(2r+s+4t)] \\
\delta f(r) \delta f(s) \delta f(t) \dd r \dd s \dd t
\end{multline*}
in direction $\delta f$.  When $T \geqslant \pi$, the first and second
order derivatives of $\F_P$ in direction $\gamma(t)=1$ for $t<\pi$ and
$0$ otherwise, vanish and evaluating the triple integral above shows
that the third-order term is proportional to $\sin\left(2(\theta-\phi)
\right)$.  If $\sin(2(\theta-\phi))\neq 0$ then either $+\gamma$ or
$-\gamma$ is a direction of increase and $f$ is a third-order saddle
point.  $\sin(2(\theta-\phi))\neq 0$ unless $\theta-\phi=0$ or $\pi/2$
(modulo $\pi$).  $\theta-\phi=0$ corresponds to the global maximum and
$\theta-\phi=\pi/2$ the global minimum of the fidelity.
\end{example}

This example shows that there are second order traps which are saddle
points.  These points are bad for optimization as escaping from a
neighborhood of such as point is difficult.  However, we can construct
examples of true traps, for which there exists a neighborhood of the
critical point from which no escape is possible.  The reason why the
second order trap above is still a saddle is the fact that the Hessian
is only negative \emph{semi}-definite.  To construct a true trap, we
therefore need to modify the example to ensure that the Hessian is
strictly negative definite, i.e., has no zero eigenvalues.

\begin{example}{\bf Non-extremal critical points with negative definite 
Hessian (traps).}  Consider $H=H_0+f(t)H_1$ as before, and choose the
target time $T=\frac{\pi}{\eps}$ and Hamiltonians and initial and target
states as follows:
\begin{align*}
  H_0 = \begin{pmatrix} 
	 1+\eps & 0 & 0 & 0 \\ 
	 0 & 1 & 0 & 0 \\
         0 & 0 & 2 & 0 \\ 
	 0 & 0 & 0 & 2 
        \end{pmatrix}, \;
  H_1 = \begin{pmatrix}
	 0 & 1 & 0 & 0 \\
	 1 & 0 & 1 & 0 \\
	 0 & 1 & 0 & b \\
	 0 & 0 & b & 0
	\end{pmatrix}, \;
  \ket{\Psi_0} = \frac{1}{\sqrt{2}} \begin{bmatrix} e^{i\theta} \\ 0 \\
				     0 \\ e^{-i\theta}\end{bmatrix}, \;
  \ket{\Psi_g} = \frac{1}{\sqrt{2}} \begin{bmatrix} e^{-iT(1+\eps)} \\ 0
				     \\ 0 \\ e^{-i2T}\end{bmatrix}.
\end{align*}
It is easy to verify that the gradient vanishes identically for $f=0$.
Hence, $f=0$ is a critical point as before.  The fidelity at $f=0$ is
$\F_P=\cos^2\theta$ and the Hessian is
\begin{align*}
 & -\underset{0<s<t<\frac{\pi}{\eps}}{\int\!\!\int} 
    [b^2 \cos^2\theta + \cos\theta \cos(\eps(s-t)-\theta)] \, 
    [\alpha(s)\beta(t)+\alpha(t)\beta(s)] \dd s \dd t\\
=&  \underset{0<s<t<\frac{\pi}{\eps}}{\int\!\!\int} \left[
   -b^2 \cos^2\theta - \cos^2\theta  \cos(\eps(s-t)) 
  - \tfrac{1}{2}\sin(2\theta) \sin(\eps(s-t)) \right]
   \, [\alpha(s)\beta(t)+\alpha(t)\beta(s)] \dd s \dd t
\end{align*}
recalling $\cos(a-b)=\cos(a)\cos(b)+\sin(a)\sin(b)$.  With the
definitions in the appendix the Hessian can be rewritten as
\begin{align*}
 -a_0 \Pi[1] -\cos^2\theta \big(\Pi[\sin(\eps\bullet)]+\Pi[\cos(\eps\bullet)] \big)  
    -\sum_{k=1}^\infty a_k \big( \Pi[\sin(2k\eps\bullet)]+\Pi[\cos(2k\eps\bullet)] \big)  
\end{align*}
with $a_0=b^2\cos^2\theta-\frac{2}{\pi}\sin(2\theta)$ and
$a_k=\frac{2\sin(2\theta)}{\pi (4k^2-1)}$ for $k>0$.  The Hessian will
be strictly negative definite if $a_k>0$ for all $k$, which is
equivalent to $b^2\cos^2\theta>\frac{2}{\pi}\sin(2\theta)>0$.  With
$b=3$, for instance, this double inequality is satisfied for all
$\theta\in(0,\frac{\pi}{3}]$, and this range of $\theta$ yields
fidelities $\F_P=\cos^2\theta \in [\frac{1}{4},1)$. Therefore, for all
of these values of $\theta$, the critical point $f\equiv 0$ corresponds
to a local maximum with $\F_P<1$, at least over any finite dimensional
subspace.

By computing the rank of the relevant Lie algebra, we can verify that
$H_0,H_1$ specifies an (exact-time) controllable system at $\eps=0$, so
that our uniform controllability result applies.  Time-$T$
controllability can therefore be ensured by choosing a sufficiently
small $\eps$, so that $H_0, H_1$ lie in the neighborhood of uniform
controllability and $T=\frac{\pi}{\eps}$ is larger than the minimal time
required for controllability.
\end{example}

This example can easily be generalized to show that traps exist for
entire families of systems, which can have any dimension.  Consider a
system of any dimension with Hamiltonians given by
\begin{align*}
\widetilde{H_0} = \begin{pmatrix}
  H_0 & \mathbf{0}\\
  \mathbf{0} & \mathbf{\ast}
  \end{pmatrix}, \quad
\widetilde{H_1} = \begin{pmatrix}
  0 & 1 & 0 & 0 & \mathbf{0}\\
  1 & \ast & \ast & 0 & \mathbf{\ast}\\
  0 & \ast & \ast & b & \mathbf{\ast}\\
  0 & 0 & b & 0 & \mathbf{0}\\
  \mathbf{0} & \mathbf{\ast} & \mathbf{\ast} & \mathbf{0} &
  \mathbf{\ast}
  \end{pmatrix}
\end{align*}
where $H_0, b$ are those in the previous example, bold characters denote
blocks, and entries with $\ast$ can be chosen arbitrarily (subject to
the result being Hermitian). For the initial and target states
\begin{align*}
  \ket{\widetilde{\Psi_0}} 
  = \begin{bmatrix}\Psi_0 \\ \mathbf{0}\end{bmatrix}, \quad
  \ket{\widetilde{\Psi_g}} 
  = \begin{bmatrix} \Psi_g \\ \mathbf{0} \end{bmatrix}
\end{align*}
the value, derivative and Hessian of this problem matches those of the 
previous example at $f\equiv 0$. Indeed, by construction, the action of 
$\widetilde{H_1}$ matches that of $H_1$ on the $e_1,e_4$ subspace, 
mapping it into the $e_2,e_3$ subspace, while evolution of both these 
invariant subspaces is the same under either $\widetilde{H_0}$ or $H_1$.
So this problem has a sub-optimal local maximum for any $\eps > 0$, as 
in the example.

\subsection{Non-constant singular controls}

Theorem~\ref{thm:1} and the previous results relied on constant controls
to show that critical points with non-extremal values of the fidelity
exist for all systems with $H_f(t)=H_0+f(t)H_1$, i.e., any choice of
$H_0$ and $H_1$, and that such critical points need not be saddle points
but can in fact be traps.  The chief motivation for this choice is that
constant controls enable analytical gradients and Hessian computations
and thus allow us to rigorously prove the existence of non-kinematical
critical points and traps.  Although even a single trap is problematic,
it is worthwhile to consider briefly how the critical point condition
$\delta\F_P/\delta f\equiv 0$ can be satisfied for a given system for
non-constant controls $f(t)$.  Let
$\bra{\Psi_B}=\overline{\bra{\Psi_g}U_f(T)\ket{\Psi_0}}\bra{\Psi_g}
U_f(T)$, then the critical point condition reads
\begin{equation}
  \Im \bra{\Psi_B} U_f(t)^\dag H_1 U_f(t) \ket{\Psi_0} \equiv 0,
\end{equation}
where $U_f(t)$ must satisfy the Schrodinger equation~(\ref{eq:SEU}),
i.e., we have a Differential-Algebraic equation system.  Note that since
$U_f(T)$ is of course not known in advance, we set
$\bra{\Psi_B}=\overline{\ip{\Psi_b}{\Psi_0}}\bra{\Psi_b}$ for some
$\bra{\Psi_b}$ and then $\ket{\Psi_g}=U_f(T)\ket{\Psi_b}$, a posteriori.
While the existence and uniqueness of solutions to such systems is in
general not trivial to ascertain, differentiating the constraint and
using the Schrodinger equation leads to a more explicit form
\begin{subequations}
\begin{align}
  0 &= \Im \bra{\Psi_B}H_1\ket{\Psi_0} \\
  0 &= \Re \bra{\Psi_B}[H_0,H_1]\ket{\Psi_0} \\
  0 &\equiv \Im\bra{\Psi_B} U_f(t)^\dag [H_0+f(t)H_1,[H_0,H_1]]U_f(t)\ket{\Psi_0}.
\end{align}
\end{subequations}
Assuming the first two equations hold, which can be thought of as a
two-dimensional constraint on $\bra{\Psi_B}$, the system can be solved,
at least locally about $t=0$, by adjoining the constraint
\begin{equation}
  f(t) = \label{eq:nonkin}
  -\frac{\Im \bra{\Psi_B} U_f(t)^\dag [H_0,[H_0,H_1]] U_f(t) \ket{\Psi_0}}
        {\Im \bra{\Psi_B} U_f(t)^\dag [H_1,[H_0,H_1]] U_f(t) \ket{\Psi_0}}
\end{equation}
to the Schrodinger equation, under the generically true condition that
the denominator does not vanish at $t=0$.  The additional, also generic,
property that $\dd f/\dd t|_{t=0} \not=0$ guarantees that the critical
point $f$ in question is not a constant function.  As the fidelity
$\F_P$ for the control $f$ is $|\ip{\Psi_b}{\Psi_0}|^2$ we can construct
critical points for which the fidelity does not take extremal values.
Eq.~(\ref{eq:nonkin}) can be solved numerically to find non-constant
non-kinematic critical points for which the fidelity assumes any desired
value.  Indeed, this was demonstrated for a four-level system in
Ref.~\refcite{wu2009}.  While some, and perhaps most, of these critical
points may be non-attractive, as appeared to be the case in the example
studied in Ref.~\refcite{wu2009}, there is no reason why \emph{all} such
critical points should be non-attractive in general.  Unfortunately, it
is difficult to prove this given only a numerical solution of
(\ref{eq:nonkin}) as we cannot calculate the Hessian exactly in this
case and prove it to be strictly negative definite.

\section{Unitary Operator Landscape}
\label{sec:unitary}

\subsection{No suboptimal attractive points (traps) over $\UU(N)$}

The set of critical points of the fidelity as a map from the unitary
group $\UU(N)$ to $\RR$ given by $\G_V(U)=\frac{1}{N}\Re\Tr[V^\dag U]$
is equal to all $U$ for which $\Re\Tr[V^\dag U A]$ vanishes for every
anti-Hermitian matrix $A\in\uu(N)$.  Since the bilinear map $\G_V$ is an
inner product over complex matrices for which the Hermitian and
anti-Hermitian matrices constitute orthogonal subspaces, the latter
condition is equivalent to $W=V^\dag U$ being Hermitian.  As $W$ is also
unitary it must be of the form $W=P_S-P_S^\perp$, where $P_S$ is a
projector onto a subspace $S$ of $\CC^N$, and $P_S^\perp$ the projector
onto the orthogonal complement of $S$.  From this we see immediately
that $\Tr(W)$ is equal to $N-2d$, where $d$ is the dimension of the
subspace $S$, and thus there are $N+1$ critical manifolds corresponding
to critical values $1-\frac{2d}{N}$ of $\G_V(U)$ for
$d\in\{0,\ldots,N\}$.  Moreover, the second derivative of $\G_V$ is
$\frac{1}{N}\Re\left[ V^{\dag} U A^2 \right]$, where $A^2$ can be any
negative semi-definite matrix.  Thus $U=\pm V$, corresponding to $W =
\pm I$ are global extrema, and all other possible $W$ have both positive
and negative eigenvalues, so that $\Re (W A^2)$ can take any value,
showing that all other critical points $U$ are saddle points.

This characterization of the critical points on $\UU(N)$, which can be
traced back to at least Ref.~\refcite{frankel1965}, provides the
motivation for the assertion in Ref.~\refcite{rabitz2005} that the
landscape for unitary operator optimization has critical points only at
these values of the fidelity, and that all critical points $U$, except
$U=\pm V$, which correspond to the global maximum and minimum of
$\G_V(U)$, respectively, are saddle points.  As in the pure-state
transfer case, to conclude this for the optimization problem over
$L^2(0,T)$ from the observations about the critical points of $\G_V$ as
a functional on $\UU(N)$, it is implicitly assumed that the solution
functional $U_f(T)$ is regular, in the differential geometric sense of
having a Jacobian of full rank everywhere over
$L^2(0,T)$. Ref.~\refcite{hsieh2008a} includes the condition that
variations of the controls can be used to generate any local variation
of $U(T)$ as part of the definition of controllability but this notion
is stronger than the usual notion of controllability.  As before, the
full-rank-everywhere assumption is problematic.  In fact for unitary
operator control problems it can \emph{never} be satisfied over any
function space that contains constant functions and thus any landscape
analysis based on this hypothesis is at best inconclusive.

For optimal control problems involving unitary operators this is even
easier to see than for pure-state optimal control problems: For any
constant control $f\equiv\mu$, the trajectory $U_{\mu}(t)$ will be a
linear combination of expressions $e^{-i \lambda_j t}$, where
$\lambda_j$ are the eigenvalues of $H_0+\mu H_1$.  The gradient of the
solution operator at this point $-iU_{\mu}(T,\tau)H_1 U_{\mu}(\tau)$ is
therefore a linear combination of functions
$e^{i(\lambda_j-\lambda_k)\tau}$ for $j,k\in\{1,\ldots,N\}$, or
equivalently of the real functions $1$,
$\cos((\lambda_j-\lambda_k)\tau)$ and $\sin((\lambda_j-\lambda_k) \tau)$
for $1\leq j<k\leq N$, of which there are only $N^2-N+1<\dim\uu(N)=N^2$.
More generally, the rank of the Jacobian when all controls are constants
$\mu_m$ is at most $N^2-N+M$ because, in the eigenbasis of $H_0+\sum
\mu_m H_m$, the diagonal elements of each $U_{\vec{\mu}}(\tau)^{\dag}H_m
U_{\vec{\mu}}(\tau)$ are constant functions.  The observation that
constant controls are singular in the sense of not being full-rank has
also been made recently in Ref.~\refcite{wu2009} but without considering
the implications for the applicability landscape results.

Also, the application of results in Ref.~\refcite{chitour2008} as
suggested in Ref.~\refcite{wu2009} is problematic as the former work,
originating in sub-Riemannian geometry requires a strictly
positive-definite running costs on the controls, which serves as a
regularizing term and is \emph{not} present in the optimal control
problems defined above, and is generally undesirable as it prevents us
from ever reaching the global maximum of the actual objective function.
Also properties generic over the infinite dimensional space of vector
fields, such as the main results of Refs~\refcite{chitour2008} or
\refcite{bonnard1997}, need not hold for any instance within the finite
dimensional subset of right-invariant vector fields which we restrict
attention to.  Thus, none of these arguments can guarantee the absence
of traps for the fidelity $\F_V$ over $L^2(0,T)$, and we shall show in
the following that traps do indeed exist.  Before we study the fidelity
over $L^2(0,T)$ further, however, it is useful to briefly consider the
critical points of $\F_V(U)$ on the Lie group in more detail, as there
are some subtle issues one should be aware of.  In particular when the
system evolution is restricted to a subgroup of $\UU(N)$ such as
$\SU(N)$, attractive critical points may emerge even on the group and it
may desirable to eliminate these by modifying the performance index.

\subsection{Attractive suboptimal critical points over $\SU(N)$}

Many problems in quantum control involve control Hamiltonians $H_1$ that
have zero trace, i.e. $iH_1 \in \su (N)$, in which case exact-time $\UU
(N)$ controllability cannot hold. Indeed, we always have $\det [U_f (t)]
= e^{- iT / N \Tr [H_0]} = e^{i \phi}$, regardless of the control
$f(t)$, so the reachable set at time $T$ is restricted to matrices in
$\SU (N)$ times the fixed phase $e^{i \phi}$. Hence, even assuming that
the system is $\UU (N)$ controllable and $V, T$ are chosen such that
$\det(V)=e^{i \phi}$ to make $V$ reachable, when considering the
critical points of the fidelity on the group, we should really consider
the critical points of $\G_V = \frac{1}{N} \Re \Tr [W]$ as a functional
over $W = V^{\dag} U \in \SU (N)$. The critical point condition is now
that $\Re\Tr(W A) = 0$ for all $A$ in the Lie algebra of trace-zero
Hermitian matrices $\su (N)$, and thus that $W$ be equal to a Hermitian
matrix $R$ plus a multiple $i \alpha$ of the identity with $\alpha$
real.  Both this and the unitary conditions are properties of the
spectrum alone, so we have in general that, $W$ is a critical point of
$\G_V$ whenever its eigenvalues are $i e^{i \theta}$ and $i
e^{-i\theta}$ with multiplicities $d$ and $N - d$, for some $d \in
\left\{ 0,\ldots,\left\lfloor N / 2 \right\rfloor\right\}$.  Finally,
the unit determinant condition on $W$ forces $e^{i \theta}$ to equal
$\exp \bigl( i \frac{\pi N}{2 ( N - 2 d )} \bigr)$ times any $|N - 2
d|^{\text{th}}$ root of unity. Note that when $N > 2$, this set of
points is larger than the set of critical points over $\UU (N)$ which
also happen to lie in $\SU (N)$. The second derivative of $\G_V ( W )$
is again $\frac{1}{N} \Re \Tr \left[ W A^2 \right]$, but now with $A \in
\su( N )$, and only the Hermitian part $R$ of $W$ contributes in this
expression, since $A^2$ is negative semi-definite. It is clear that if
$R$ has eigenvalues of both signs, implying $N \geq 3$, then the
corresponding critical point $W$ is a second order saddle. Otherwise $d
= 0$, and the critical point is a local maximum or minimum if $R$ is a
positive or negative multiple of the identity, while in case $R = 0$, we
can find curves $\alpha i e^{A x}$ for which $\G_V$ does not vanish to
third order in $x$, so $W$ is a saddle point. This characterizes the
attractive critical points of $\G_V(U)$ as those $U$ of the form $e^{i
\phi} V$ for some $N^{\text{th}}$ root of unity $e^{i \phi}$
having positive real part, and continuity of the solution operator $U_f
(T)$ then implies:

\begin{theorem}
For any controllable system $H=H_0+f(t)H_1$ with $\Tr(H_1)=0$, and any
$F=\cos(2\pi k/N)$ with $k=1,\ldots,\lceil N/4\rceil-1$, there exists an
open subset $\mathcal{N}$ of $L^2(0,T)$ such that any monotonically
increasing local optimization algorithm started with $f\in W$ will never
exceed a fidelity of $F$.
\end{theorem}

\subsection{Elimination of traps by optimization over $\PU(N)$}

These results show that that for optimization of $\SU(N)$, there are
traps, i.e., attractive critical points of $\G_V(U)$ with $\G_V(U)<1$,
and this is relevant for optimization problems over $L^2(0,T)$ with
$\Tr[H_1]=0$.  As the attractive critical points differ only by a global
phase factor $e^{i\theta}$ from the target gate $V$, however, and we
usually do not care about the global phase of a gate, these traps are
not really a problem in practice, provided the convergence condition of
the algorithm is more sophisticated than checking if $1-\F_V(U_f(T))$ is
less than a certain tolerance.  However, if we do not care about the
global phase of the gate $V$, a better choice of the performance index
to be optimized would be
\begin{equation}
  \G_V':U \mapsto\frac{1}{N^2}|\Tr(V^\dag U)|^2, 
\end{equation}
which is equivalent to optimization over the projective unitary group
$\PU(N)$.  To see why it is a better choice, note that the critical
point condition for $\G_V'$ is $\Re\left[\bar{\gamma}\Tr(W A)\right]=0$
for every $A\in\su(N)$ with $W=V^\dag U$ as before, and $\gamma=\Tr[W]$.
This is satisfied only when $\bar{\gamma}W=\alpha iI+R$ for some real
$\alpha$ and Hermitian $R$ with all eigenvalues of the same magnitude.
The second derivative at these points is $2\Re\Tr[RA^2]+2|\Tr[W A]|^2$,
the first term can only fail to ever be positive for any $A\in\su(N)$ if
$R$ is a non-negative multiple of the identity, or $N=2$.  In the first
case, either $\gamma=0$ or $W=e^{i\theta}I$, which corresponds to
$\G_V'$ attaining its minimal or maximal values respectively.
Otherwise, $R$ has eigenvalues of opposite sign, so the second
derivative's first term must vanish for every $A\in\su(N)$, and if its
second term always vanishes then $U$ must have equal eigenvalues, thus
maximizes $\G_V'$.  This modification of the objective function is
preferable in practice as it eliminates root-of-unity traps.

\subsection{Non-global maxima for optimization over $L^2(0,T)$}

The previous section shows that suboptimal attractive critical points do
arise for optimization of $\SU(N)$, and that this case is relevant for
quantum control problems, but these can easily be eliminated by changing
the performance index.  We now turn our attention to the existence of
sub-optimal attractive critical points for the actual optimization
problem of interest, i.e., the problem of maximizing the fidelity $\F_V$
over $L^2(0,T)$.  In particular we are interested in whether there are
suboptimal attractive critical points, i.e., points $U_f(T)$ for which
the gradient of $\F_V$ vanishes, the Hessian is negative definite, and
$\F_V(U_f(T))<1$ for systems that are time-$T$ controllable.

\begin{example}{\bf Traps for unitary operation optimzation problem. }
Consider the control problem specified by:
\begin{align*}
  H_0 = 
  \begin{pmatrix} 1 + \eps & 0 & 0\\ 0 & 1 & 0\\ 0 & 0 & 2 \end{pmatrix},
  \quad 
  H_1 = 
  \begin{pmatrix} a & 1 & 0\\ 1 & b & 1\\ 0 & 1 & c \end{pmatrix},
  \quad 
  V^{\dag} = 
  \begin{pmatrix} 
   e^{i \phi} & 0 & 0\\ 0 & i e^{i \gamma} & 0\\ 0 & 0 & i e^{-i\gamma}
  \end{pmatrix} e^{i T H_0} 
\end{align*}
with $T=\frac{\pi}{\eps}$.  At $f \equiv 0$ the fidelity is $\F_V =
 \cos\phi$ and the gradient and Hessian are 
\begin{align*}
  \nabla \F_V   &= \frac{1}{3}\int_0^{T} g(\tau) \alpha(\tau) \, d\tau \\
  \nabla^2 \F_V &= \frac{1}{3}\underset{0<\sigma<\tau<T}{\int\!\!\int} B(\sigma,\tau)
 \left[ \alpha (\tau) \beta (\sigma)+ \beta(\tau) \alpha(\sigma) \right] \dd \sigma \dd \tau
\end{align*}
where $g(\tau)= a\sin(\phi) + (b+c)\cos(\gamma)$ and
\begin{align*}
  B(\sigma,\tau) 
 &= -a^2\cos\phi-(c^2-b^2)\sin\gamma -\cos(\phi+\eps(\sigma-\tau))+ \sin(\gamma-\eps(\sigma-\tau))\\
 &= -a^2\cos\phi-(c^2-b^2)\sin\gamma
 -(\cos\gamma-\sin\phi)\sin(\eps(\sigma-\tau))
 -(\cos\phi-\sin\gamma)\cos(\eps(\sigma-\tau)).
\end{align*}
The last line follows from basic trigonometric identities. Using the
definitions in \ref{appendix} we can rewrite the Hessian as
\begin{multline*}
 -\frac{1}{3} w\Pi[1]-\frac{1}{3} (\cos\gamma-\sin\phi) \sum_{k=1}^{\infty}
  \frac{4}{\pi(4k^2-1)}\big(\Pi[\sin(2k\eps\bullet)]+\Pi[\cos(2k\eps\bullet)]\big)\\
 - \frac{1}{3} (\cos\phi - \sin\gamma)
  \big(\Pi[\cos(\eps\bullet)]+\Pi[\sin(\eps\bullet)]\big)
\end{multline*}
where
$w=a^2\cos\phi+(c^2-b^2)\sin\gamma+\tfrac{4}{\pi}(\cos\gamma-\sin\phi)$.
A sufficient set of conditions for $\F_V(f)$ to have a local maximum at
$f=0$ is that $\nabla \F_V(f)=0$ and $\nabla^2 \F_V(f)$ strictly
negative definite for $f=0$.  This is equivalent to
\begin{subequations} 
\label{egconsts}
\begin{align}
  {0} &= a\sin\phi + (b+c)\cos\gamma, \\
  0 &<\tfrac{4}{\pi}(\cos\gamma - \sin\phi)+ a^2\cos\phi + (c^2 - b^2) \sin\gamma,\\ 
  0 &<\cos\gamma-\sin\phi,\\
  0 &\le \cos\phi - \sin\gamma.
\end{align}
\end{subequations}
There are many solutions to this set of equations/inequalities.  For 
instance, suppose $(b+c)/a>0$ and $\gamma$ in the 4th quadrant. Then
$\cos\gamma>0$ and $\sin\gamma<0$.  The first equality then requires
$\sin\phi = -(b+c)/a \cos\gamma<0$ but we can choose $\phi$ to be 
in the 4th quadrant so that $\cos\phi>0$.  Then we have
$\sin\gamma<0<\cos\phi$ satisfying the last inequality, and
$\sin\phi<0<\cos\gamma$ satisfying the third inequality. 
This guarantees that the first two terms in the second inequality are
positive.  The third term will be non-negative if $b^2\ge c^2$.
\end{example}

\subsection{Variable time non-global maxima}

In typical quantum control problems, the target time $T$ is fixed, but
$T$ can be allowed to vary.  If the target time $T$ of control problems
is not fixed then there is some ambiguity as to how proximity between
control fields is to be measured when these lie in different function
spaces $L^2(0,T)$.  We can resolve this difficulty by optimizing over
Hamiltonians of the form $H(t)=\ell H_0+f(t)H_1$ with $T$ fixed, since
propagating $H(t)$ up to time $T$ is equivalent to propagating
$H_0+\frac{1}{\ell}f(t/\ell)H_1$ up to time $\ell T$.  The expressions
for the gradient~(\ref{eq:grd}) and Hessian~(\ref{eq:hess}) derived
earlier from the perturbative expansion can be adapted to this variable
time optimization framework simply by replacing all instances of $\Delta
f H_1=H_{f+\Delta f}-H_f$ by $\Delta\ell H_0+\Delta f H_1$.  In
particular, at $f \equiv 0$ the gradient of the objective functional
simplifies to
\begin{equation} 
  \label{eq:grd2}
  (\ell_a \oplus \alpha) \mapsto
  -i\int_0^T U_f(T,\tau)H_1\Delta f(\tau)U_f(\tau)\dd\tau 
  -iU(T)TH_0\Delta\ell 
\end{equation}
while its Hessian is the bilinear map on $\RR \oplus L^2(0,T)$ given by
\begin{multline}
  \label{eq:hess2}
  (\ell_a \oplus \alpha, \ell_b \oplus \beta) \mapsto 
   -\!\!\!\underset{0<\sigma<\tau<T}{\int\int}\!\!\! 
    U_f(T,\tau) H_1 U_f(\tau,\sigma) H_1 U_f(\sigma) 
    \left[\alpha(\tau)\beta(\sigma)+\beta(\tau)\alpha(\sigma)\right] 
    \dd\sigma\dd\tau \\
   - U(T)\int_0^T 
    \left([U_f(\tau)^{\dag}H_1U_f(\tau),H_0]\tau +
          H_0U(\tau)^{\dag}H_1U_f(\tau) T \right) 
    [\alpha(\tau)\ell_b+\beta(\tau)\ell_a]\dd\tau \\
    - U_f(T)T^2 H_0^2 \ell_a \ell_b.
\end{multline}

\begin{example}
In this framework consider the unitary control problem specified by
\begin{align*}
  H_0 = \begin{pmatrix}
  1 + \varepsilon & 0 & 0 & 0\\
  0 & 1 & 0 & 0\\
  0 & 0 & 2 & 0\\
  0 & 0 & 0 & 3
\end{pmatrix}, \quad
 H_1 = \begin{pmatrix}
  0 & 1 & 0 & 0\\
  1 & b & 1 & 0\\
  0 & 1 & c & g\\
  0 & 0 & g & d
\end{pmatrix}, 
\end{align*}
\begin{align*}
V^{\dag} = \begin{pmatrix}
  r & 0 & 0 & s e^{i \left( \theta - \gamma \right)}\\
  0 & - i & 0 & 0\\
  0 & 0 & e^{i \phi} & 0\\
  - s e^{i \gamma} & 0 & 0 & r e^{i \theta}
\end{pmatrix}
e^{i T H_0} 
\end{align*}
with $T=\frac{\pi}{\eps}$, where a fortiori $s=\sqrt{1-r^2}$, and we
shall set $b=\frac{5}{2}(1-\sqrt{3})$, $c=\sqrt{3}-3$, $d=3$,
$g=\sqrt{\frac{3}{71}(19+12\sqrt{3})}$,
$r=\frac{2}{9}(3-\sqrt{3})\approx 0.28$, $\phi=\frac{\pi}{3}$,
$\theta=-\frac{\pi}{3}$ with $\gamma$ arbitrary.  With these parameters,
the gradient at $f \equiv 0$, $\ell=1$ always vanishes, and the Hessian
evaluates to
\begin{align*}
 -C T^2 \dd\ell^2 - 0.80\ \Pi[1] 
 & - 1.60\ \big( \Pi[\cos]+\Pi[\sin] \big) 
   - 0.28\ \big( \Pi[\cos(\eps\bullet)]+\Pi[\sin(\eps\bullet)] \big)\\
 & -\sum_{k=1}^{\infty} \frac{4}{\pi(4k^2-1)} 
    \big(\Pi[\sin(2k\eps\bullet)]+\Pi[\cos(2k\eps\bullet)] \big),
\end{align*}
which is always positive definite, since $C$, although dependent on
$\eps$, must be at least $3.5$, and the fidelity $\F_U$ here is 
$\frac{1}{48}(33-2\sqrt{3})\approx 62\%$. 
\end{example}

\section{Conclusion}

We have revisited the control landscape for several classes of canonical
quantum control problems, in particular pure-state transfer and unitary
gate optimization problems.

Although the pure-state transfer fidelity as a function over the unitary
group only has two types of critical points, corresponding to either the
global minimum $0$ or the global maximum $1$, detailed analysis shows
that the class of critical points for the actual optimization problem
over functions in $L^2(0,T)$ is larger than the set of kinematic
critical points.  In particular, for any bilinear control system
$H_f(t)=H_0+f(t)H_1$ and any fixed target time $T$, there exist pairs of
initial and target states, such that some $f\equiv \mbox{\rm const.}$ is
a critical point of the system, and we can achieve any value of the
fidelity between $0$ and $1$ for such critical points.  Moreover, while
these critical points are not expected to be attractive in most cases,
we have presented examples of systems with suboptimal critical points at
which the Hessian is negative definite with infinite rank, showing that
traps \emph{do} exist for such problems, and the fidelity at these traps
can take many values, unlike the fidelity for kinematic critical points,
which is limited to 0 or 1 for pure-state optimization problems.

For the problem of unitary operator optimization we demonstrated that
there are no traps when the fidelity is taken to be a functional over
the unitary group.  Although there are critical manifolds on which the
fidelity takes values between $-1$ (global minimum) and $+1$ (global
maximum), all of these critical points are indeed saddle points, but the
situation changes when the analysis is restricted to the special unitary
group $\SU(N)$.  This case may appear artificial but it is actually
highly relevant for quantum control, as many quantum control problems
involve control Hamiltonians that have zero-trace, and hence we have no
global phase control.  Although the system may be $\UU(N)$ controllable
when $\Tr[H_1]=0$, if the target time is fixed, so is the global phase.
There are more critical points in this case than for $\UU(N)$, and some
of these are attractive critical points at which the fidelity assumes
values $<1$, i.e., traps.  More careful analysis shows that these traps
correspond to solutions $U=e^{i\theta}V$, where $e^{i\theta}$ is a root
of unity, which immediately shows that these solutions have fidelity
$<1$, according to the standard definition of the fidelity.
Nonetheless, these solutions are perfectly adequate for most practical
purposes when the global phase of an operator is not important.
Moreover, this problem can be avoided entirely simply by modifying the
performance index to reflect the fact that we do not care about the
global phase.  Optimizating over the projective unitary group $\PU(N)$
there are indeed no traps but again the situation is more complicated
for actual optimization problem over controls in $L^2(0,T)$.  Not only
are there critical points $f$ for which the fidelity assumes critical
values other than those permitted over $\UU(N)$ or $\SU(N)$, but
examples again show that the Hessian at these critical points can be
infinite-rank negative definite, implying that they are locally
attractive on any finite-dimensional subspace and therefore traps.  The
results can even be extended to the case where the target time $T$ is
variable, again proving the existence of traps even in this case.

The specific examples of traps constructed here prove an important
theoretical point about the existence of non-kinematical critical points
and traps in the control landscape, but perhaps more importantly, the
results raise many questions about the control landscape.  Can explicit
examples be constructed of non-constant controls which can be proven to
be traps in the sense that the gradient vanishes and the Hessian is
negative definite with infinite rank?  How common are these traps?  Are
there problems for which no such traps exist?  When iterative methods
are employed to find optimal controls, what is the domain of attraction
of the traps in the landscape and how does it depend on the algorithm
used?  The landscape also depends on the domain, i.e., the space of
controls.  Here we assumed $f\in L^2(0,T)$ but other function spaces can
be considered.  In practice the controls are usually restricted to a
finite-dimensional subspace of $L^2(0,T)$ as we do not have infinite
time and frequency resolution.  The zero-control traps are interesting
in this context as any finite-dimensional subspace of $L^2(0,T)$ will
contain such controls and therefore traps, but in general the control
landscapes may look very different for different subspaces.

\section*{Acknowledgement}

We thank Thomas Schulte-Herbr\"uggen, John V. Leahy, Alex Pechen, David
Tannor and Barry Sanders for helpful discussions and suggestions, and
the Kavli Institute of Theoretical Physics at UCSB and Banff
International Research Station (Canada) for enabling these.  This work
was supported by funding from EPSRC Advanced Research Fellowship
EP/D07192X/1, CASE studentship CASE/CNA/07/47, Hitachi and NSF Grant
PHY05-51164.

\appendix
\section{Definition of Operators $S$ and $C$}
\label{appendix}

We define the operators $C$ and $S$ via their action on functions
$\alpha,\beta\in L^2(0,T)$.  The functions $\{\sigma_0, \sigma_k^c,
\sigma_k^s: k\in \NN\}$ with $\sigma_0=\sqrt{\tfrac{1}{T}}$,
$\sigma_k^c=\sqrt{\tfrac{2}{T}}\cos(2k\omega t)$ and
$\sigma_k^s=\sqrt{\tfrac{2}{T}}\sin(2k\omega t)$ form an orthonormal
basis for $L^2(0,T)$.  For the operator $S$ defined by
\begin{equation}
\label{eq:S}
\ip{\beta}{S\alpha}_{L^2}
  = \!\!\!\underset{0<\sigma<\tau<T}{\int\int}\!\!\! \sin(\omega(\sigma-\tau)) 
     [\alpha(\tau)\beta(\sigma)+\beta(\tau)\alpha(\sigma)]\dd\sigma\dd\tau 
\end{equation}
we observe that the Fourier basis for $L^2(0,T)$ is an eigenbasis of $S$
by verifying that
\begin{align*}
 \ip{\sigma^s_k}{S \sigma^s_\ell} = \frac{2T}{\pi(4k^2-1)}
 \delta_{k\ell}, \quad
 \ip{\sigma^c_k}{S \sigma^c_\ell} = \frac{2T}{\pi(4k^2-1)} \delta_{k\ell}, \quad
 \ip{\sigma^c_k}{S \sigma^s_\ell} = 0.
\end{align*}
With the notation defined earlier, where $\Pi[f(\bullet)]$ defines a
projection operator mapping any $\alpha\in L^2(0,T)$ onto
$f(\bullet)\int_0^T \alpha(s)f(s)\dd s$ in $L^2(0,T)$, we can thus write
\begin{equation}
\begin{split}
  S &= -\frac{2T}{\pi} \Pi[\sigma_0] 
       + \sum_{k=1}^\infty \frac{2T}{\pi (4k^2-1)}(\Pi[\sigma_k^c] + \Pi[\sigma_k^s]) \\
    &= -\frac{4}{\pi} \Pi[\sigma_0] 
       + \sum_{k=1}^\infty \frac{4}{\pi (4k^2-1)}(\Pi[\cos(2k\omega\bullet)] + \Pi[\sin(2k\omega\bullet)]) 
\end{split}
\end{equation}
where the unnormalized projection in the second line is obtained by
simply multiplying the coefficients by $\tfrac{2}{T}$.  Notice that all
eigenvalues are non-zero, and except for the eigenvalue corresponding to
$\sigma_0^c$, positive.  For the operator $C$ defined similarly
\begin{equation}
\label{eq:C}
  \ip{\beta}{C\alpha}_{L^2}
  = \!\!\!\underset{0<s<t<T}{\int\int}\!\!\! \cos(\omega(s-t)) 
     [\alpha(t)\beta(s)+\beta(t)\alpha(s)]\dd s\dd t.
\end{equation}
the Fourier basis of $L^2(0,T)$ is not an eigenbasis 
\begin{gather*}
 \ip{\sigma^c_k}{S \sigma^s_\ell} = 0, \quad
 \ip{\sigma_0}{C \sigma_0} = \frac{T}{2} \frac{8}{\pi^2}, \quad
 \ip{\sigma^c_k}{C \sigma_0} = \frac{T}{2} \frac{\sqrt{2}}{\pi} \frac{4}{\pi(4k^2-1)} \\ 
 \ip{\sigma^s_k}{C \sigma^s_\ell} = \frac{T}{2} \frac{4}{\pi(4k^2-1)}\frac{4}{\pi(4\ell^2-1)}, \quad
 \ip{\sigma^c_k}{C \sigma^c_\ell} = \frac{T}{2}
 \frac{8k}{\pi(4k^2-1)}\frac{8\ell}{\pi(4\ell^2-1)}, \nonumber \\
 \hspace{0.84\textwidth} k,\ell>0
\end{gather*}
but it can be verified that $C$ is diagonalized by choosing $\alpha$ and $\beta$
\begin{align*}
 \alpha(t) &= \sum_{k=1}^\infty \frac{8k}{\pi(4k^2-1)} \sigma^s_k(t) \to \cos(\omega t) \\
 \beta(t)  &= \frac{2}{\pi} -\sum_{k=1}^\infty \frac{4}{\pi(4k^2-1)} \sigma^s_k(t) \to \sin(\omega t) 
\end{align*}
and thus we can write
\begin{equation}
 C = \Pi[\sin(\omega \bullet)] + \Pi[\cos(\omega\bullet)].
\end{equation}
\bibliography{landscape,unenc}

\end{document}